\documentclass[letterpaper, 10 pt, conference]{ieeeconf}  

\IEEEoverridecommandlockouts                              

\usepackage{amsmath,amssymb,amsfonts}
\usepackage{graphicx}
\usepackage{comment}
\usepackage{circuitikz}
\usepackage[dvipsnames]{xcolor}
\allowdisplaybreaks
\usepackage{layout}
\usepackage{cite}
\usepackage{hyperref}
\usepackage{url}            
\usepackage{pgfplots}
\pgfplotsset{compat=newest}
\usepackage{tikz}
\usepackage{bm}
\usepackage{adjustbox}
\usepackage{MnSymbol}
\pgfplotsset{compat=1.18}
\usetikzlibrary{arrows.meta, decorations.markings, shapes.misc, shapes.geometric}

\newtheorem{theorem}{Theorem}
\newtheorem{assumption}{Assumption}
\newtheorem{definition}{Definition}
\newtheorem{lemma}{Lemma}

\title{\LARGE \bf
On the Existence of Quadratic Control Lyapunov Functions for Koopman-Operator based Bilinear Systems
}

\author{Sami Leon Noel Aziz Hanna$^1$, Nicolas Hoischen$^1$, Sandra Hirche$^1$, Armin Lederer$^2$
\thanks{*This work was supported by the DAAD program Konrad Zuse Schools of Excellence in Artificial Intelligence, sponsored by the Federal Ministry of Education and Research, and by the European Union’s Horizon Europe innovation action program under grant agreement No. 101093822,
”SeaClear2.0” and the Deutsche Forschungsgemeinschaft (DFG) as part of the Research Unit "Active Learning for Systems and Control (ALeSCo)" under project number 535860958.}
\thanks{$^1$Chair of Information-oriented Control, School of Computation, Information and Technology, Technical University of Munich, Germany {\tt\small [sami.noel, nicolas.hoischen, hirche]@tum.de}.}
\thanks{$^2$Department of Electrical and Computer Engineering, National University of Singapore {\tt\small armin.lederer@nus.edu.sg}}
}

\begin{document}
    \maketitle
    \thispagestyle{empty}
    \pagestyle{empty}
    
    \begin{abstract}             
        Koopman operator-based methods enable data-driven bilinear representations of unknown nonlinear control systems. Accurate representations often demand significantly higher dimensions than the original system, making control design challenging. Control Lyapunov Functions (CLFs) are widely used for controller synthesis, with quadratic CLF candidates being the most common due to their simplicity. Yet, we show that this class is highly restrictive, especially when the state dimension is large: under mild conditions, their existence implies stabilizability of the bilinear system by a constant input---that is, the control remains fixed over time. We establish this result by formulating a quadratically constrained quadratic program (QCQP) that exactly characterizes valid CLFs. Since QCQPs are NP-hard, we propose a convex semidefinite relaxation that offers a sufficient validity condition. For single-input systems, we prove that a quadratic CLF requires constant control stabilizability, and empirically demonstrate that this extends to high-dimensional multi-input systems in many cases.
    \end{abstract}

    \section{Introduction}\label{sec:intro}

When first-principles approaches for system identification reach their limits, such as in neuroprosthetics \cite{wright2016} or glucose dynamics \cite{wayne2012}, data-driven methods become indispensable for constructing suitable models. Recently, Koopman operator–theoretic techniques for nonlinear system identification have attracted considerable attention within the control community due to their ability to derive data-driven representations of unknown nonlinear autonomous systems in a linear but higher (possibly infinite)- dimensional space~\cite{koopman31}. Similarly, under mild conditions~\cite{goswami2020}, nonlinear, control-affine systems can be expressed with bilinear models, motivating the development of approximate, finite-dimensional representations through methods such as extended dynamic mode decomposition~\cite{otto2024}, kernel-based approaches~\cite{bevanda2024}, and neural networks~\cite{zheng2024}. Such models offer a simple structure for formulating control laws for a given task, such as stability. 

One of the most common techniques for designing stabilizing control laws are control Lyapunov functions (CLFs). The existence of such functions is known to be a sufficient and necessary condition for stabilizability~\cite{freeman2008}. Moreover, CLFs offer a principled approach to account for model uncertainties or exogenous disturbances—a crucial advantage within learning-based approaches that can quantify model uncertainty (see \cite{strasser2025kernelbased} for recent advancements in the Koopman operator framework). 

While CLFs offer many beneficial properties for control design with lifted bilinear Koopman models, synthesizing CLFs themselves remains challenging in general. This has led to a restriction to CLFs from simple function classes, among which quadratic CLFs are by far the most common in current literature \cite{strasser2024, huang2018, narasingam2023}. The synthesis of such CLFs only requires finding a positive definite matrix, enabling the search for them using convex semidefinite programs (SDPs) \cite{huang2018} or linear matrix inequalities \cite{strasser2024}. The conditions for the feasibility of the proposed optimization problems are unclear, resulting in a limited understanding of when the proposed approaches in the literature are successful.
Particularly in high-dimensional settings, such as those in Koopman operator learning, it is unknown primarily when a bilinear representation actually admits a quadratic CLF. 

To the best of our knowledge, only the requirement of constant control that asymptotically stabilizes the system is known as a straightforward requirement to ensure the existence of a quadratic CLF in this scenario \cite{andrieu2013}. Constant control, i.e., applying the same control input at all times, transforms the bilinear representation into a linear system such that the existence of quadratic CLFs becomes obvious. However, these finite-dimensional bilinear models obtained from Koopman operator-theoretic methods are typically just approximations of nonlinear systems, for which constant control stabilizability is highly questionable.

We show that the existence of a quadratic CLF for a bilinear system is equivalent to the sufficiency of constant control for stabilization under mild assumptions. For this, we first derive a quadratically constrained quadratic program (QCQP) for stability analysis by exploiting the structure of bilinear systems and quadratic CLFs. This optimization problem yields identical solutions as its relaxation into an SDP under certain assumptions. We demonstrate that these assumptions hold for bilinear systems with scalar control inputs and empirically show that they are commonly satisfied in systems where the state dimension significantly exceeds that of the control inputs. Finally, we prove that the SDP-based stability condition holds if and only if constant control can stabilize the bilinear system. Due to the restrictiveness of the requirement of constant control stabilization, our results raise critical questions regarding the scalability of existing quadratic CLF approaches to high-dimensional bilinear Koopman embeddings.

The remainder of this paper is structured as follows. Section~\ref{sec:background} provides background on the Koopman operator, followed by a formal problem statement in Section \ref{sec:problem}. Section~\ref{sec:constant_control_almost_necessary} works out existence conditions for quadratic CLFs and presents our main result, namely the requirement of stabilizability by constant control for the existence of a quadratic CLF for single-input systems, and, under mild conditions, for high-dimensional multi-input systems. Finally, Section~\ref{sec:conclusion} concludes the paper. 

    \section{Koopman Operator Background}\label{sec:background}
The Koopman operator\footnote{\textbf{Notation:} Matrices and vectors are written in bold. We write $\mathbb{R}$ for the real numbers and $\mathbb{C}$ for the complex numbers, respectively. $\mathbb{C}^-$ for the left half of the complex plane, $\mathbb{N}_{\leq N}$ for the natural numbers less than or equal to $N$, $\mathbb{Z}$ for integers and $\mathcal{H}$ denotes a Hilbert space. $\bm Q\succ \bm 0$ means the matrix $\bm Q$ is positive definite and $\bm Q\succeq \bm 0$ positive semi-definite (similarly $\prec, \preceq$ for the negative analogue). The same notation is applied to functions, e.g., $V\succ 0$ means the function $V$ is positive definite. We write $\operatorname{sym}(\bm A)$ for the symmetric component of matrix $\bm A$, i.e. $\operatorname{sym}(\bm A) = \frac{1}{2}(\bm A + \bm A^\top)$.} $\mathcal{K}_t:\mathcal{H}\rightarrow\mathcal{H}$ \cite{koopman31} is defined over the composition
\begin{align}
    \mathcal{K}_t\varphi := \varphi\circ \bm F^t(\bm x(\bm 0)) 
\end{align}
of an observable functional $\varphi\in\mathcal{H}:\mathcal{X}\rightarrow\mathbb{C}$ with a smooth and Lipschitz flow \cite{bevandar2021} 
\begin{align}
    \bm F^t(\bm x(0)) := \bm x(0) + \int_{\tau=0}^t \bm f(\bm x(\tau))\ \mathrm{d}\tau
\end{align}
of the autonomous dynamical system $\dot{\bm x} = \bm f(\bm x)$, $\bm x\in\mathbb{R}^n$ with nonlinear drift $\bm f:\mathbb{R}^n\rightarrow\mathbb{R}^n$. We assume forward invariance of $\mathcal{X}\subset \mathbb{R}^n$ to ensure that both the Koopman operator and observable functions are well-defined. In this context, the Koopman operator acts on an observable function by propagating it forward in time within the Hilbert space $\mathcal{H}$, analogous to how the flow evolves the state $\bm{x}$ in the state space $\mathcal{X}$. Notably, this temporal evolution is linear, even when the underlying dynamics of the state are nonlinear. Similarly, for a control-affine system of the form
\begin{align}\label{eq:nonlinear_contr_affine_sys}
    \dot {\bm x} = \bm f(\bm x) + \sum_{i=1}^m\bm g_i(\bm x) u_i, \quad \bm u\in\mathbb{R}^m,
\end{align}
with control vector field $\bm g_i:\mathbb{R}^n\rightarrow\mathbb{R}^n$, the Koopman operator yields a bilinear, possibly infinite dimensional representation under the invariant subspace condition, given in \cite[Theorem 1]{goswami2017}. The evolution of an observable functional $\varphi$ is given by the partial differential equation
\begin{align}
    \dot \varphi = \mathcal{L}_{\bm f}\varphi + \sum_{i=1}^m\mathcal{L}_{\bm g_i}\varphi u_i,
\end{align}
where $\mathcal{L}:\mathcal{H}\rightarrow\mathcal{H}$ denotes the Lie-derivative, which is a linear but infinite dimensional operator. To obtain a classic control system with a finite state dimension, the method of Extended Dynamic Mode Decomposition (EDMD) projects these operators onto a finite dimensional subspace of so called observable functionals $\varphi_i$, collected in a dictionary $\mathcal{D}$. In some cases, this subspace can be chosen to be invariant under their temporal evolution, but this is typically not the case, especially if this subspace is chosen without explicit knowledge about the system dynamics. In those cases, one can remain on this subspace by means of orthogonal projection and derive representation error estimates \cite{schaller_2023}. The finite dimensional, matrix-valued operator approximators for the drift ($\bm A$) and control terms ($\bm B_i$) are then derived as follows: Assuming that data $\{\bm{x}_j^{\bar{\bm{u}}}, \dot{\bm{x}}_j^{\bar{\bm{u}}}\}_{j=1}^{J_{\Bar{u}}}$ of the system \eqref{eq:nonlinear_contr_affine_sys}, excited by constant control inputs $\bar{\bm{u}}$, are available, construct the vector of observables  
$\Hat{\bm{\Phi}} = [\varphi_1, \ldots, \varphi_N]^\top$ as well as the regression observable matrix 
\begin{align}
    \mathbb{X}^{\bar{\bm{u}}} = [\Hat{\bm{\Phi}}(\bm{x}_1^{\bar{\bm{u}}}), \ldots, \Hat{\bm \Phi}(\bm x_{J_{\Bar{u}}}^{\bar{\bm{u}}})]
\end{align}
and regression observable derivative matrix
\begin{align}
    \mathbb{Y}^{\bar{\bm{u}}} = [\nabla \Hat{\bm{\Phi}}(\bm{x}_1^{\bar{\bm{u}}})^\top\cdot\dot{\bm{x}}_1^{\bar{\bm{u}}}, \ldots, \nabla \Hat{\bm{\Phi}}(\bm x_{J_{\Bar{u}}}^{\bar{\bm{u}}})^\top\cdot\dot{\bm{x}}_{J_{\Bar{u}}}^{\bar{\bm{u}}}].
\end{align}
The matrices $\bm{A}$ and $\bm{B}_i$ are then identified by solving the regression problems \cite{strasser2024}
\begin{align}
    \bm{A} &= \arg\min_{\hat{\bm{A}} \in \mathbb{R}^{N \times N}} 
        \|\mathbb{Y}^{\bm{0}} - \hat{\bm{A}} \mathbb{X}^{\bm{0}}\|_F, \\
    \bm{B}_i &= \arg\min_{\hat{\bm{B}}_i \in \mathbb{R}^{N \times N}} 
        \|\mathbb{Y}^{\bm{e}_i} - \hat{\bm{B}}_i \mathbb{X}^{\bm{e}_i}\|_F - \bm{A},
\end{align}
where $\bm{\Bar{u}}$ is set to the zero input and basis vectors $\bm{e}_i \in \mathbb{R}^m$ in order to excite the system in all input directions and $\|\cdot\|_F$ denotes the Frobenius norm. For the lifted state $\bm{z} = \bm {\Hat{\Phi}}(\bm x)$, the resulting finite-dimensional Koopman-based model takes the bilinear form
\begin{align}\label{eq:bilinear_form}
    \dot{\bm{z}} = \bm{A}\bm{z} + \sum_{i=1}^m \bm{B}_i \bm{z} u_i, 
    \quad \bm z \in \mathbb{R}^N, \quad \bm{u} \in \mathbb{R}^m.
\end{align}
In general, the learned matrices do not possess any particular structure, and their dimensionality depends on the number of observables chosen in $\mathcal{D}$. For simple system dynamics, low-dimensional, handcrafted observables may achieve satisfactory learning accuracy \cite{strasser2024}. However, for practically relevant systems with complex dynamics and little prior knowledge about suitable observable functions, a large ($N\gg n$) dictionary of generic observables, e.g., monomials, is typically required for minor learning errors \cite{bruder2021}. Recently, kernel based methods were introduced, which model observables in a reproducing kernel hilbert space, which mitigates the problem of handcrafting observables, as they can be written as linear kernel combinations \cite{bevanda2024}.  
    \section{Problem Formulation}\label{sec:problem}
Before analyzing bilinear systems, we formally introduce the concept of stability.
\begin{definition}[Stability]\label{def:stability}
    For a given feedback control law, the equilibrium $\bm z^* = \bm 0$ of the system \eqref{eq:bilinear_form} is asymptotically stable, if for every $\epsilon>0\ \exists\ \delta>0$, such that
    \begin{align}
        \|\bm z(t_0)\|<\delta \Rightarrow \|\bm z(t)\|<\epsilon, \ \forall t\geq t_0,
    \end{align}
    and $\delta(\epsilon,t_0)>0$ can be chosen, such that
    \begin{align}
        \|\bm z(t_0)\|<\delta \Rightarrow \lim_{t\rightarrow \infty}\|\bm z(t)\| = 0.
    \end{align}
    If the system admits an asymptotically stabilizing controller, it is said to be asymptotically stabilizable.
\end{definition}
To synthesize a stabilizing controller, we can use CLFs as they are both sufficient and necessary for this task \cite{freeman2008}. A CLF is formally defined to satisfy the following properties.
\begin{definition}[Control Lyapunov function]\label{def:CLF}
    A differen– tiable function $V:\mathcal{X}\mapsto\mathbb{R}$ is called control Lyapunov function, if it satisfies
    \begin{align}\label{eq:CLF}
        V(\bm x) \succ 0, \quad \inf_{\bm u\in\mathbb{R}^m} \dot{V}(\bm x, \bm u) \prec 0.
    \end{align}    
\end{definition}
Note, that we do not assume any restriction on the control input. In this work, we focus on homogeneous bilinear forms, where homogeneity is defined as follows.
\begin{definition}[Homogeneity]\label{def:homogeneity}
    A function $f\!:\mathbb{R}^N \rightarrow \mathbb{R}$ is homogeneous if there exists an integer $p\in\mathbb{Z}$ such that \looseness=-1
    \begin{align}\label{eq:homogeneity}
        f(t\bm z) = t^pf(\bm z), \quad \forall \bm z\in\mathbb{R}^N, \ t\in\mathbb{R}.
    \end{align}
    Vector fields are considered to be homogeneous if each entry satisfies \eqref{eq:homogeneity}.
\end{definition}
Our main objective is to find sufficient and necessary conditions under which the homogeneous bilinear system \eqref{eq:bilinear_form} admits a control Lyapunov function (CLF) in quadratic form
\begin{align}\label{eq:quadratic_form}
    \mathcal{V}_\text{quad} := \{V_{\bm Q}\succ 0 \mid V_{\bm Q}(\bm z) \mapsto \frac{1}{2}\bm z^\top \bm Q\bm z\}.
\end{align}
More specifically, we seek to find conditions for the existence of a symmetric positive definite matrix $\bm Q$ such that the corresponding quadratic candidate function $V_{\bm Q}$ satisfies the requirements given in \eqref{eq:CLF}. 
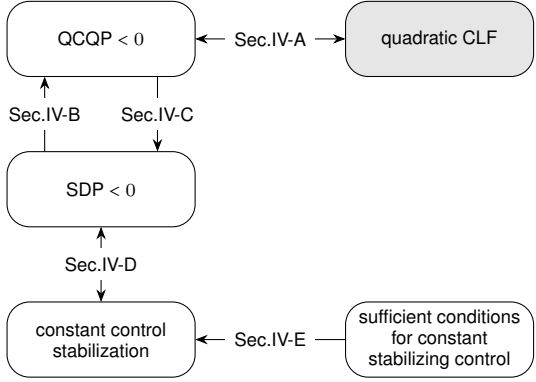
\begin{figure}[t]
    \centering
    \vspace{.2cm}
        \begin{circuitikz}
        \tikzstyle{every node}=[font=\scriptsize]
        \draw[rounded corners=0.3cm]  (11.25,19) rectangle  node {\scriptsize\sffamily $\text{QCQP}<0$} (13.75,18);
        \draw[rounded corners=0.3cm]  (11.25,17) rectangle  node {\scriptsize\sffamily $\text{SDP}<0$} (13.75,16);
        \draw[rounded corners=0.3cm]  (11.25,15) rectangle  node[align=center] {\scriptsize\sffamily constant control\\ \scriptsize\sffamily stabilization } (13.75,14);
        \draw [<-, >=Stealth] (11.75,18) -- (11.75,17) node[pos=0.5, fill=white]{\sffamily Sec.\ref{sec:SDR}};
        \draw [->, >=Stealth] (13.25,18) -- (13.25,17) node[pos=0.5, fill=white]{\sffamily Sec.\ref{sec:SDP-exactness}};

        \draw[rounded corners=0.3cm]  (15.75,15) rectangle  node[align=center] {\scriptsize\sffamily sufficient conditions \\ \scriptsize\sffamily for constant \\ \scriptsize\sffamily stabilizing control} (18.25,14);
        \draw[rounded corners=0.3cm, fill = gray!20]  (15.75,19) rectangle  node {\scriptsize\sffamily quadratic CLF} (18.25,18);

        \draw [<-, >=Stealth] (13.75,14.5) -- (15.75,14.5)node[pos=0.5, fill=white]{\sffamily Sec.\ref{sec:suff_cond}};
        \draw [<->, >=Stealth] (13.75,18.5) -- (15.75,18.5)node[pos=0.5, fill=white]{\sffamily Sec.\ref{sec:existence_verification}};
        \draw [<->, >=Stealth] (12.5,16) -- (12.5,15)node[pos=0.5, fill=white]{\sffamily Sec.\ref{sec:constant_control}};
    \end{circuitikz}
    \caption{Relation between the SDP, QCQP, existence of a quadratic CLF, and constant control. While the existence of a constant, stabilizing controller is sufficient for the existence of a quadratic CLF, necessity holds if the SDP and QCQP attain the same optimal value. We also give sufficient conditions for the existence of constant stabilizing control.}
    \label{fig:flowchart}    
\end{figure}

    \section{Constant Control as an Almost-Necessary Condition in High Dimensions}\label{sec:constant_control_almost_necessary}

In this section, we show the equivalence between the existence of a quadratic CLF and constant stabilizing control for bilinear systems using the sequence of arguments outlined in Figure~\ref{fig:flowchart}. 
To this end, we propose formulating a quadratically constrained quadratic program (QCQP) in Section~\ref{sec:existence_verification}, in which the negativity of the objective value serves as both a necessary and sufficient condition for the existence of a quadratic control Lyapunov function. Owing to the inherent non-convexity of the problem, we employ a semidefinite programming relaxation in Section~\ref{sec:SDR} that yields a sufficient condition for the existence of a CLF. In Section~\ref{sec:SDP-exactness}, we show the equivalence between the QCQP and its SDP relaxation for single-input systems and demonstrate empirically that it also holds for high-dimensional multi-input systems. We derive constant control as a sufficient and necessary condition for the existence of a quadratic CLF in Section~\ref{sec:constant_control} by assuming this equivalence. Finally, we give sufficient conditions for the existence of constant stabilizing control and consequently for the existence of a quadratic CLF in Section~\ref{sec:suff_cond}. 

\subsection{CLF-Condition through a QCQP}\label{sec:existence_verification}
Given bilinear system \eqref{eq:bilinear_form}, the Lyapunov descent condition for the quadratic function class \eqref{eq:quadratic_form} is given by 
\begin{align}\label{eq:descent_condition}
    \dot V_{\bm Q} = \bm z^\top \bm Q \bm A\bm z + \sum_{i=1}^m \bm z^\top \bm Q \bm B_i\bm zu_i \prec 0.
\end{align}
The bilinear structure separates the Lyapunov derivative into an autonomous and a control-dependent term. Notably, both are homogeneous quadratic expressions depending on the Lyapunov matrix $\bm Q$. The zero-level set of each summand in the second term defines a manifold
\begin{align}\label{eq:generic_zero_lvl_set}
    \mathcal{M}^i_{V_{\bm Q}} := \{\bm z\in\mathbb{R}^n\backslash \{\bm 0\}\mid\bm z^\top \bm Q\bm B_i\bm z = 0\},
\end{align}
such that on the intersection
\begin{align}
    \mathcal{M}_{V_{\bm Q}} := \bigcap_{i=1}^m \mathcal{M}^i_{V_{\bm Q}},
\end{align}
no controller can enforce Lyapunov descent. This set is important for a structural CLF search program as it manifests the regions in state space, where the autonomous term $\bm z^\top \bm Q \bm A\bm z$ must provide negativity by itself to ensure descent condition \eqref{eq:descent_condition}. A valid candidate $V_{\bm Q}\succ 0$, therefore either satisfies the negativity of the autonomous term when control authority is lost, that is
\begin{align}\label{eq:negativity_condition}
    \bm z^\top\bm Q\bm A \bm z < 0 \ \forall \bm z\in\mathcal{M}_{V_{\bm Q}},
\end{align}
or ensures control authority globally, i.e.
\begin{equation}\label{eq:emptiness_condition}
    \mathcal{M}_{V_{\bm Q}} = \emptyset.
\end{equation}
A high number of unstable eigenvalues in $\bm A$ introduces many different directions in the set $\{\bm z\in\mathbb{R}^N\mid \bm z^\top\bm Q\bm A \bm z \geq 0\}$, which makes the task of separating it from $\mathcal{M}_{V_{\bm Q}}$ difficult.
Notice that both expressions $\bm z^\top \bm Q\bm B_i\bm z$ and $\bm z^\top \bm Q\bm A\bm z$ satisfy the homogeneity property \ref{def:homogeneity} with $p=2$ and therefore do not change sign after scaling the vector $\bm z$. Therefore, without loss of generality, the state is restricted to the $N-1$ dimensional unit sphere
\begin{align}\label{eq:unit_sphere_constr}
    \mathbb{S}^{N-1}:=\{\bm z\in\mathbb{R}^N\mid \|\bm z\| = 1\},
\end{align}
to keep the problem bounded. Based on this insight, we incorporate \eqref{eq:negativity_condition} and \eqref{eq:emptiness_condition} into the QCQP
\begin{subequations}\label{eq:QCQP}
\begin{align}
    [\star] := &\max_{\bm z\in\mathbb{S}^{N-1}} \bm z^\top \bm Q\bm A\bm z  \\
    &\text{s.t. } \bm z^\top \bm Q \bm B_i\bm z = 0, \quad i = 1,...,m. \label{eq:QCQP_constraints}
\end{align}
\end{subequations}
Constraints \eqref{eq:QCQP_constraints} restrict the optimization problem to the the space $\mathcal{M}_{V_{\bm Q}}$.
Equivalently to \eqref{eq:negativity_condition} and \eqref{eq:emptiness_condition}, we can thus verify whether a positive definite matrix $\bm Q$ parameterizes a valid CLF by the negativity of the optimal value defined as
\begin{align}
    J^*_\text{QCQP} = \begin{cases}
        [\star] & \text{if \eqref{eq:QCQP} is feasible}\\
        -\infty & \text{otherwise.}
    \end{cases}
\end{align}

\begin{lemma}\label{lem:QCQP_validation}
    Matrix $\bm Q$ parameterizes a valid CLF $V_{\bm Q}$ if and only if $J^*_\text{QCQP}(\bm Q) < 0$.
\end{lemma}
\begin{proof}    
    We start with sufficiency: If $J^*_\text{QCQP}(\bm Q) < 0$, then problem \eqref{eq:QCQP} is either infeasible or returns a negative value. Infeasibility means there exists no state vector $\bm z$ satisfying the equality constraints \eqref{eq:QCQP_constraints}, therefore \eqref{eq:emptiness_condition} is fulfilled and the controller can shape the Lyapunov derivative globally. If the problem returns a negative value, then \eqref{eq:negativity_condition} is fulfilled and negativity of the Lyapunov derivative is guaranteed in $\mathcal{M}_{V_{\bm Q}}$. In both cases, $\bm Q$ parameterizes a valid Lyapunov function. For necessity, assume $V_{\bm Q}$ is a CLF and $J^*_\text{QCQP}(\bm Q) \geq 0$. This means problem \eqref{eq:QCQP} is feasible and there exists a state $\bm z$ in the set $\mathcal{M}_{V_{\bm Q}}$ such that $\dot V_{\bm Q} \geq 0$, which is a contradiction.
\end{proof}

Lemma \ref{lem:QCQP_validation} provides an exact validity condition for a quadratic CLF matrix $\bm Q$. Even though QCQPs \eqref{eq:QCQP} are NP hard to solve \cite{argue2023}, they can be relaxed by an SDP for convexity, as we show in the next section. 

\subsection{Semidefinite Relaxation}\label{sec:SDR}
To reformulate the QCQP into an SDP, we reformulate the quadratic expressions using the trace operator and its rotational invariance. For the objective of \eqref{eq:QCQP}, this yields
\begin{align}
    \bm z^\top \bm Q\bm A \bm z &= \bm z^\top \operatorname{sym}(\bm Q\bm A) \bm z = \operatorname{trace}(\operatorname{sym}(\bm Q\bm A) \bm Z)
\end{align}
where $\bm Z = \bm z\bm z^\top \succeq \bm 0$ and $\operatorname{rank}(\bm Z)=1$. The same transformation can be applied to the equality constraints \eqref{eq:QCQP_constraints} and the unit sphere constraint \eqref{eq:unit_sphere_constr} in its quadratic form
\begin{align}
    \bm z^\top \bm z = 1\qquad
    \Leftrightarrow \qquad\operatorname{trace}(\bm z\bm z^\top) = \operatorname{trace}(\bm Z) = 1.
\end{align}
To obtain a convex SDP, we drop the requirement $\operatorname{rank}(\bm Z)=1$, such that the relaxed problem becomes
\begin{subequations}\label{eq:SDP}
\begin{align}
    [\ostar] = &\max_{\bm Z\succeq \bm 0}\operatorname{trace}((\operatorname{sym}(\bm Q\bm A) \bm Z) \quad \\
    &\text{s.t. }\operatorname{trace}(\bm Z) = 1\label{eq:SDP_con1}\\
    &\quad~~\!\operatorname{trace}((\operatorname{sym}(\bm Q\bm B_i) \bm Z) = 0, \quad i = 1,...,m.\label{eq:SDP_con2}
\end{align}
\end{subequations}
Similarly to before, we define the optimal value as 
\begin{align}\label{eq:SDP_sol}
    J^*_\text{SDP}(\bm Q) = 
    \begin{cases}
        [\ostar] & \text{if \eqref{eq:SDP} is feasible}\\
        -\infty & \text{otherwise.}
    \end{cases}
\end{align}
By ignoring the requirement $\operatorname{rank}(\bm Z) = 1$, the \textit{nonconvex} feasible set of matrices for the QCQP \eqref{eq:QCQP} is enlarged to a \textit{convex} set for the SDP \eqref{eq:SDP} as illustrated in Figure \ref{fig:feasible_domain}. 
\tikzset{
  mystar/.style={shape=star,star points=5,minimum size=#1,star point ratio=3,rounded corners=#1/10}
}
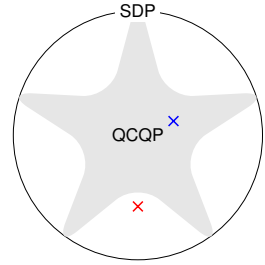
\begin{figure}[t]
    \centering
    \begin{adjustbox}{width = .2\textwidth, center}
    \begin{circuitikz}
        \tikzstyle{every node}=[font=\Large]
        \draw [color=black] (6,16) circle (3.5cm);
        \node[color = black, fill = gray!20, mystar={7.9cm}] at (6,16) {};
        \node[draw, cross out, very thick, minimum size=6pt, blue] at (7,16.4) {};
        \node[draw, cross out, very thick, minimum size=6pt, red] at (6,14) {};
        \node [font=\Large, color=black, fill = white] at (6,19.5) {\sffamily SDP};
        \node [font=\Large, color=black] at (6,16) {\sffamily QCQP};
    \end{circuitikz}
    \end{adjustbox}
    \vspace{-0.6cm}
    \caption{Illustration of feasible domains for the SDP and QCQP, the blue cross illustrates a rank 1 solution of the SDP, while the red cross illustrates a solution of rank larger than one. Note that this illustration only serves to show the set enlargement as 
    rank 1 matrices lie at the border of a semidefinite cone.}    
    \label{fig:feasible_domain}
\end{figure}
As the SDP solution \eqref{eq:SDP_sol} provides an upper bound for the QCQP solution, we immediately obtain a sufficient condition for a quadratic CLF.
\begin{lemma}\label{lem:SDP_validation}
Matrix $\bm Q\succ\bm 0$ parameterizes a valid CLF $V_{\bm Q}$ if $J^*_\text{SDP}(\bm Q)<0$.
\end{lemma}
\begin{proof}
    Due to the feasible set enlargement, the maximum of the SDP problem \eqref{eq:SDP} serves as an upper bound to the maximum of the QCQP problem \eqref{eq:QCQP}. Further, for the same reason, infeasibility of the SDP guarantees infeasibility of the QCQP, in this case $J_\text{QCQP}^*= J_\text{SDP}^* = -\infty$. The objectives therefore satisfy $J_\text{QCQP}^* \leq J_\text{SDP}^*$, and negativity of $J_\text{SDP}^*$ guarantees negativity of $J_\text{QCQP}^*$. The remainder of the proof follows from Lemma \ref{lem:QCQP_validation}.
\end{proof}

\subsection{SDP-Exactness}\label{sec:SDP-exactness}
In the previous section, we showed that SDP \eqref{eq:SDP} serves as an upper bound for QCQP \eqref{eq:QCQP}. In certain cases, both problems obtain the same solution.

\begin{lemma}[SDP-exactness]\label{lem:SDP_exact}
    Problem \eqref{eq:SDP} recovers the solution to QCQP \eqref{eq:QCQP}, if the solution $\bm Z^*$ has rank 1. In this case, $J^*_\text{SDP}<0$ is a necessary condition for matrix $\bm Q$ to parameterize a valid CLF.
\end{lemma}

\begin{proof}
    Any positive semidefinite, rank 1 matrix $\bm Z$ can be written as the outer product of two vectors $\bm Z = \bm z\bm z^\top$. After solving \eqref{eq:SDP} and obtaining the solution $\bm Z^*$, we can perform the resubstitution for the objective
    \begin{align}
        \operatorname{trace}(\bm Q\bm A\bm Z^*) = \operatorname{trace}(\bm z^{*\top}\bm Q\bm A\bm z^*)
    \end{align}
    and similarly for the constraints \eqref{eq:QCQP_constraints}. The objectives therefore satisfy $J_\text{QCQP}^* = J_\text{SDP}^*$. As $J^*_\text{QCQP}<0$ is a necessary condition for matrix $\bm Q$ to parameterize a CLF by Lemma~\ref{lem:QCQP_validation}, $J_\text{SDP}^*<0$ is necessary as well. The vector $\bm z^*$ recovers the global optimum of QCQP \eqref{eq:QCQP}, concluding the proof.
\end{proof}
In the following, we will call problems \eqref{eq:SDP} SDP-exact if they recover the solution of the QCQP \eqref{eq:QCQP}. Due to Lemmas~\ref{lem:SDP_validation} and \ref{lem:SDP_exact}, $J^*_\text{SDP}(\bm Q)<0$ is a necessary and sufficient condition for $\bm{Q}$ to parameterize a valid CLF when SDP-exactness holds, which illustrates the significance of this property. Importantly, it can be straightforwardly shown to always hold for bilinear systems with scalar control input.
\begin{lemma}\label{lem:single_input}
    Problem \eqref{eq:SDP} is SDP-exact if $m=1$.
\end{lemma}
\begin{proof}
    Let $\bm Z^*$ denote the extremal point of \eqref{eq:SDP} with rank $r\in\mathbb{N}_{\leq N}$ and $n_c$ the number of constraints. Then Pataki's Lemma \cite{pataki1998} states that 
     \begin{align}\label{eq:Pataki}
         \frac{r(r+1)}{2} \leq n_c.
     \end{align}     
    Problem \eqref{eq:SDP} includes $n_c = m+1=2$ many constraints. Solving \eqref{eq:Pataki} for $r$ with $n_c=2$ yields $r\leq -1/2\pm\sqrt{17}/2 $, which implies $r=1$ as it is a natural number.
\end{proof}

When the system has multiple control inputs, SDP-exactness is not guaranteed in general, but it is not obvious under which circumstances it is not satisfied. Therefore, we investigate the SDP \eqref{eq:SDP} in more detail. For this, note that in the absence of \eqref{eq:SDP_con2}, the maximizer of \eqref{eq:SDP} becomes $\bm{Z}=\bm{v}_1\bm{v}_1^T$, where $\bm{v}_1$ is the unit-length eigenvector corresponding to the largest eigenvalue of $\operatorname{sym}(\bm{Q}\bm{A})$ and $\bm{Z}$ (by construction). This fact immediately follows from the relationship between traces and eigenvalues of matrices. When we have $m$ control inputs, we get $m$ additional constraints, which restrict the admissible values for $\bm{z}$ in \eqref{eq:SDP}. If $\bm{v}_1$ is still among the admissible vectors, it remains the optimizer and SDP-exactness is ensured. It is straightforward to see that for $N\approx m$, this is difficult to satisfy due to the small number of admissible values that remain for $\bm{z}$. However, in the scenario of Koopman-based learned bilinear systems \eqref{eq:bilinear_form}, we often have $N\gg m$ since large dictionaries are required without access to manually designed observables as discussed in Section \ref{sec:background}. Then, $N-m\approx N$ degrees of freedom remain for $\bm{z}$, such that we hypothesize that the optimization problem \eqref{eq:SDP} behaves almost as if the constraints \eqref{eq:SDP_con2} did not exist and SDP-exactness is generally retained. This illustrates that SDP-exactness is tightly coupled to the size $N$ of the dictionary for Koopman-based learning of bilinear systems.
\looseness=-1

To support our hypothesis, we empirically analyze SDP-exactness in a simulation. To this end, we solve the SDP \eqref{eq:SDP} 20 times for each $m\in\{2,5,8\}$ and $N\in\{10,...,100\}$, where we sample random matrices $\bm A, \bm B_i, \bm P\in\mathbb{R}^{N\times N}$ and set $\bm Q = \bm P\bm P^\top$ to obtain a positive definite matrix. Moreover, we manipulate matrix $\bm A$ to possess only one unstable eigenvalue, such that the least restrictive, non-trivial scenario is considered for \eqref{eq:SDP}. Finally, if the problem is feasible, we extract the rank of $\bm Z^*$ numerically (with a tolerance of $10^{-7}$ for its eigenvalues); otherwise, we resample and solve again. The results of this numerical evaluation are illustrated in Figure \ref{fig:rank_Z_plot}. It can be clearly observed that the average rank of the maximizer $\bm{Z}^*$ approaches $1$ for increasing $N$. While the effect is slower for larger numbers of control inputs $m$, this fully aligns with the previous discussion. Thus, SDP-exactness is consistently observed for large $N$. 

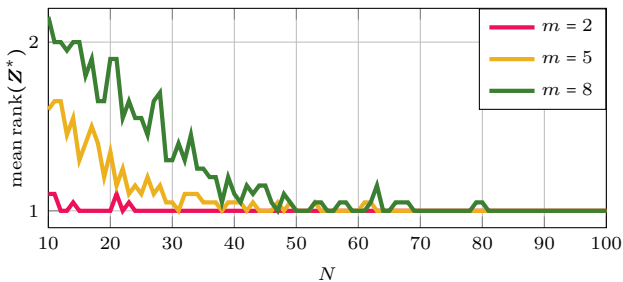
\begin{figure}[t]
    \vspace{.2cm}
    \definecolor{mycolor1}{rgb}{0.92941,0.69412,0.12549}
    \pgfplotsset{width=9cm, compat = 1.18, 
	height = 4.5cm, grid= major, 
	legend cell align = left, ticklabel style = {font=\scriptsize},
	every axis label/.append style={font=\scriptsize},
	legend style = {font=\scriptsize},
    }
    \begin{tikzpicture}
	\begin{axis}[
		grid=both,
		xmin=10, xmax=100,
		ymin=0.9, ymax=2.2,
		xtick distance=10,
		ytick distance=1,
		ylabel=$\operatorname{mean}\operatorname{rank}(\bm Z^*)$, xlabel=$N$,
		set layers=standard,
		legend style={font=\scriptsize, at={(1,1)},anchor=north east, row sep=2pt},
		ylabel shift = -6 pt]
				
		\addplot[ultra thick,OrangeRed] table[x=n,y=m2]{rank_plot.txt};
		\addlegendentry{$m=2$}
				
		\addplot[ultra thick,mycolor1] table[x=n,y=m5]{rank_plot.txt};
		\addlegendentry{$m=5$}
				
		\addplot[ultra thick,OliveGreen] table[x=n,y=m8]{rank_plot.txt};
		\addlegendentry{$m=8$}
				
        \end{axis}
    \end{tikzpicture}
    \vspace{-0.3cm}
    \caption{Averaged ranks of $20$ recovered SDP-solutions for each state and control dimension ($N\in\{10,...,100\},\ m\in\{2,5,8\}$). For $N>80$, for each $N$ and $m$, $\operatorname{rank}(\bm Z^*) = 1$ across all 20 runs. Note that in the operator-theoretic framework, $N$ typically exceeds the displayed dimensions.}
    \label{fig:rank_Z_plot}
\end{figure}

\subsection{Constant Control}\label{sec:constant_control}
To complete our derivations, we make the following assumption justified by the results in  Section \ref{sec:SDP-exactness}.

\begin{assumption}\label{ass:SDP_exactness}
    The SDP-solution satisfies SDP-exactness. 
\end{assumption}

Under this assumption, it suffices to analyze the convex SDP domain for existence conditions of a quadratic CLF. We present our main result in the following theorem.
\begin{theorem}\label{thm:constant_control}
    Under Assumption \ref{ass:SDP_exactness}, a valid quadratic CLF exists if and only if there exists a constant stabilizing controller.
\end{theorem}
\begin{proof}
Sufficiency follows from the fact that a constant stabilizing controller turns the bilinear system into an asymptotically stable linear system which always admits a quadratic CLF \cite[Theorem 4.6]{khalil2002}. For necessity, assume that $\bm Q$ parameterizes a CLF. Under Assumption \eqref{ass:SDP_exactness}, this implies that $J_\text{SDP}^*(\bm Q) < 0$ due to Lemma \ref{lem:SDP_exact}.
Now, define the feasibility program
\begin{subequations}\label{eq:SDP_feasibility_problem}
\begin{align}
    &\text{find}\ \bm Z\succeq \bm 0, \ \bm Z\neq \bm 0 \\
    &\text{s.t. }\operatorname{trace}(\operatorname{sym}(\bm Q\bm A)\bm Z) \geq 0\\
    &\quad~~\!\operatorname{trace}(\operatorname{sym}(\bm Q\bm B_i)\bm Z) = 0, \ i=1,...,m
\end{align}
\end{subequations}
whose infeasibility follows from $J^*_\text{SDP}(\bm Q)<0$. Due to the infeasibility of  \eqref{eq:SDP_feasibility_problem}, the semidefinite version of Farkas Lemma \cite[Lemma 6.3.3]{lovasz2003} guarantees the existence of $c_1,...,c_m\in\mathbb{R}$, such that
\begin{align}\label{eq:infeasibility_certificate}
    -\operatorname{sym}(\bm Q\bm A)+\sum_{i=1}^m \operatorname{sym}(\bm Q \bm B_i)c_i \succ \bm 0.
\end{align}
Finally, we define constant controllers $u_i = -c_i$, which satisfy the Lyapunov descent condition \eqref{eq:descent_condition} due to \eqref{eq:infeasibility_certificate}. Thus, stability of the closed-loop system under the constant control input $u_i=-c_i$ follows from \cite[Theorem 4.2]{khalil2002}.
\end{proof}

While Assumption \ref{ass:SDP_exactness} is a crucial condition for this result, it is empirically observed
to hold for bilinear systems with large state space in Section \ref{sec:SDP-exactness}. Such bilinear models are typically approximations of nonlinear systems when they are obtained using Koopman operator-based methods with large dictionaries. Since these nonlinear systems are usually not stabilizable by constant control inputs, an accurate approximation using a bilinear Koopman operator-based model can also not be expected to be stabilizable by constant control. Empirically, we also observe that constant stabilizing controllers for bilinear systems are rare.
Thus, Theorem~\ref{thm:constant_control} raises crucial questions about the existence of quadratic CLFs for these high-dimensional Koopman operator-based bilinear models. Note that this insight is not a contradiction to prior work: Our results do not show that constant control stabilization is a requirement for quadratic CLFs when dealing with low-dimensional bilinear models, which are commonly the focus in numerical evaluations \cite{strasser2024}.

\subsection{Sufficient Condition for Constant Control}\label{sec:suff_cond}

Determining whether a constant stabilizing controller exists is an NP-hard problem~\cite{toker1995}, and a detailed analysis lies beyond the scope of this paper. For the single-input case, we refer the reader to~\cite{luesink1989, elliott2009}. The next Lemma presents two sufficient condition ensuring the existence of constant stabilizing controllers, and consequently, of a quadratic CLF. It should be emphasized that the goal is not to employ these constant controllers; rather, their existence is a necessary prerequisite for constructing a quadratic CLF.

\begin{lemma}\label{lem:sufficient_existence_condition}
    The system \eqref{eq:bilinear_form} admits a constant stabilizing controller, if at least one of the following two conditions hold:
    \begin{itemize}
        \item [1.] Matrix $\bm A$ only possesses eigenvalues with negative real part.
        \item [2.] At least one of the matrices $\bm B_i$ only possesses eigenvalues with purely negative or positive, real part. 
    \end{itemize}
\end{lemma}
\begin{proof}
    The first case is easy to verify, because the drift is already asymptotically stable and the constant controller $\bm u = \bm 0$ leads to an asymptotically stable system. For the second case, we define scalar multipliers 
    \begin{align}
        c_i = \quad \begin{cases}
            \epsilon_i & \text{if } \mathfrak{Re}(\lambda_{\max}(\bm B_i))<0 \\
            -\epsilon_i & \text{if } \mathfrak{Re}(\lambda_{\max}(-\bm B_i))<0\\
            0 & \text{otherwise,}
        \end{cases}
    \end{align}
    where $\epsilon_i>0$ and $\mathfrak{Re}(\cdot)$ denotes the real part of a complex number and $\lambda_{\max}(\cdot)$ maps to the eigenvalue with the largest real part. This way, the linear combination $\sum_{i=1}^m \bm B_ic_i$ only possesses eigenvalues with negative real part and a matrix $\bm Q\succ \bm 0$ exists such that $\sum_{i=1}^m \bm z^\top \bm Q\bm B_i \bm z c_i \prec 0$. Let $u_i = c_i\Bar{u}, \Bar{u}\in\mathbb{R}$ and recall the Lyapunov descent condition \eqref{eq:descent_condition} for system \eqref{eq:bilinear_form}, then
    \begin{align}
        \dot V_{\bm Q} = \ &\bm z^\top (\bm Q \bm A + \Bar{u}\bm Q\sum_{i=1}^m \bm B_ic_i)\bm z\\
        \leq \ &\lambda_{\max}(\operatorname{sym}(\bm Q\bm A))\|\bm z\|^2\notag \\ 
        + \ &\Bar{u}\lambda_{\max}(\operatorname{sym}(\bm Q\sum_{i=1}^m \bm B_ic_i))\|\bm z\|^2 \overset{!}{\prec} 0.
    \end{align}
    We thus require
    \begin{align}
        \lambda_{\max}(\operatorname{sym}(\bm Q\bm A)) + \Bar{u}\lambda_{\max}(\operatorname{sym}(\sum_{i=1}^m \bm Q\bm c_i\bm B_i)) < 0.
    \end{align}
    For any $\delta>0$, the constant controller 
    \begin{align}
        \Bar{u} = -\frac{\lambda_{\max}(\operatorname{sym}(\bm Q\bm A))+\delta}{\lambda_{\max}(\operatorname{sym}(\sum_{i=1}^m \bm Q\bm c_i\bm B_i))}
    \end{align}
    asymptotically stabilizes system \eqref{eq:bilinear_form}, concluding the proof.\looseness = -1
\end{proof}

    \section{Conclusion}\label{sec:conclusion}
This work highlights the inherent limitations of quadratic Lyapunov functions when applied to learned bilinear and homogeneous Koopman models. We show that constant control is an almost necessary condition for a high-dimensional homogeneous bilinear system to admit a quadratic CLF, and we establish that this condition always holds for single-input systems, regardless of their dimension. These findings suggest that quadratic CLFs may be inadequate for stabilizing high-dimensional bilinear models unless additional structural constraints are imposed during the learning process to ensure compatibility with constant-control conditions.
    
    \bibliographystyle{IEEEtran}
    \bibliography{IEEEabrv, references}

\end{document}